%% file: appendix.tex
\documentclass[acmsmall,nonacm]{acmart}
\usepackage{algorithm}
\usepackage[noend]{algpseudocode}
\usepackage{xcolor, pifont}
\usepackage{bm}
\usepackage{balance}
\usepackage{caption}
\usepackage{float} 
\usepackage{url}
\usepackage{subcaption}
\usepackage{tikz}
\usepackage{amsmath}
\usepackage{amsthm}
\usetikzlibrary{tikzmark}

\tikzset{mycircled/.style={circle,draw,inner sep=0.1em,line width=0.04em}}

\newtheorem{invariant}{Invariant}
\newtheorem{lemma}{Lemma}
\newtheorem{property}{Property}

\AtBeginDocument{%
  }

\renewcommand\footnotetextcopyrightpermission[1]{}
\begin{document}

\author{
    {Xincheng Yang} \\
    Illinois Institute of Technology \\
    \and
    {Kyle Hale}\\
    Oregon State University
}

\title{The Consistency Correctness in CoPPar Tree}

\renewcommand{\shortauthors}{X. Yang et al.}

\begin{abstract}
  This article serves as a supplementary document to the CoPPar Tree paper (ISPDC 2025), presenting a detailed correctness proof for the CoPPar architecture.
\end{abstract}

\maketitle

\settopmatter{authorsperrow=0}

\section{Introduction}
In the main paper, we introduced Write-order broadcast in CoPPar tree, which uses a global order to ensure consistency. However, in the main paper, we did not formally prove the correctness of our structure. Therefore, we provide the relevant proof in this supplemental document.

In this work, we primarily adopt definitions from two papers: one from Maurice Herlihy's Linearizability paper and the other from Lamport's Sequential consistency paper. Besides, We also adopted Kfir Lev-Ari's OSC paper to unify these two consistency definitions. To maintain readability, we will restate key definitions from these works. In addition, we also define the Composition Order Cycle and prove the correctness of CoPPar Tree.

In the definition from the Linearizability paper, objects satisfy locality: the system's history is linearizable if and only if, for every object x, the history of x is linearizable.

However, in many coordination services, by allowing clients to read from
a local server, they sacrifice the consistency level and thus do not satisfy
Linearizability. For example, ZooKeeper guarantees Linearizable writes and
sequential reads. While this still provides strong consistency guarantees, it
means that it does not satisfy the locality property of Linearizability: two
sequentially consistent systems combined together do not guarantee sequential consistency
as a whole. 

We refer to this as a composition problem, and thus independent, sequentially consistent
systems are not composable. In particular, by allowing stale reads, the process
order from different processors may form cyclic dependencies. These cyclic
dependencies prevent us from obtaining a global sequential order for all
operations.

Our work, CoPPar tree, builds on the theorem and correctness results from this
paper to solve the composition problem and provide a sequentially composable
model through our structure. 

\section{Model and Notation}
\input{notation}

\section{Definition of the COC Problem}
\input{coc}

\section{CoPPar Tree Correctness Proof}
\input{correctness}

\balance
\bibliographystyle{ACM-Reference-Format}
\bibliography{xincheng}

\end{document}

%% file: notation.tex
We use a standard replicated shared memory model to demonstrate the correctness of our structure. In this model, a set of objects $\Theta$ is shared and concurrently accessed by a collection of sequential threads $\Pi$ from processes \emph{P}. An object has a name, also referred to as a key, and a corresponding value. An operation \emph{op} is used for reading or writing the value of an object. The execution of an operation consists of an \emph{invocation} and a \emph{response} event, which we denote as $inv(op)$ and $res(op)$. A history \emph{H} is a finite sequence of operations. 

We reuse the definitions from the Linearizability paper~\cite{herlihy1990linearizability} and assume that all histories are well-formed in the context of this chapter:

\begin{quotation}
A \emph{process subhistory}, $H|P$, consists of all events in the history \emph{H} that are associated with the process \emph{P}. Two histories \emph{H} and \emph{H}' are \emph{equivalent} if, for every process \emph{P}, $H|P = H'|P$.
A \emph{sequential history} meets two requirements: (1) the first event in the history \emph{H} is an invocation, and (2) each invocation, except possibly the last one, is immediately followed by a corresponding response. Every response is directly followed by a corresponding invocation. If a history is not sequential, it is called a \emph{concurrent history}. A \emph{single-object} subhistory refers to a subsequence where all events are related to a single object, denoted by $H|x$ when the object is $x$. A process subhistory $H|P$ is the subsequence of all events in history \emph{H} that involve process \emph{P}. A history \emph{H} is considered \emph{well-formed} if every process subhistory $H|P$ is sequential. A complete history, denoted as \emph{complete(H)}, is a history where the last (pending) invocation is removed. A set \emph{S} of histories is considered \emph{prefix-closed} if, for any history \emph{H} in \emph{S}, every prefix of \emph{H} is also included in \emph{S}. A \emph{sequential specification} for an object is a prefix-closed set of single-object sequential histories associated with that object. A sequential history \emph{H} is considered \emph{legal} if each object subhistory $H|x$ belongs to the sequential specification of the corresponding object \emph{x}.
\end{quotation}

We also reuse an irreflexive partial order \( <_H \) from the Linearizability paper~\cite{herlihy1990linearizability} to represent the order of operations, where \( op_0 <_H op_1 \) denotes that \( \text{res}(op_0) \) precedes \( \text{inv}(op_1) \) in history \emph{H}. Two operations are \emph{concurrent} if there is no \( <_H \) relationship between them.

We then reproduce the definition of \emph{ordered sequential consistency} from \(\text{OSC}\)~\cite{lev2017composing} paper:

\begin{quotation}
A history \( H \) is OSC if it can be extended (by appending zero or more response events) to some history \( H' \), and there is a sequential specification \( S \) of \emph{complete(H')}, such that:
\begin{enumerate}
    \item L1 (Sequential specification): \( \forall x \in \Theta \), \( S|x \) belongs to the sequential specification of \( x \).
    \item L2 (Process order): For two operations \( op_0 \) and \( op_1 \), if \( \exists P \in \Pi \) and \( op_0 <_{H|P} op_1 \), then \( op_0 <_S op_1 \).
    \item L3 (Real-time order): There exists a subset of object operations \( A \) such that \( \forall x \in \Theta \), for an operation \( op_0 \) (not necessarily in \( A \)) and an operation \( op_1 \in A \), such that \( op_0, op_1 \in H|x \), if \( op_0 <_H op_1 \) then \( op_0 <_S op_1 \).
\end{enumerate}
\end{quotation}

By the definition of consistency, both Linearizability and Sequential Consistency are special cases of OSC. 

We achieve sequential consistency by satisfying L1 and L2, which represent a global order and a process order, respectively~\cite{lamport1979make}. Linearizability is achieved when the subset \( A \) includes all operations, which causes L3 to imply \( <_H \subset <_S \). Under this condition, L1 and L2 ensure that \(\text{complete}(H')\) is equivalent to a legal sequential history \( S \). Therefore, linearizability follows from its definition~\cite{herlihy1990linearizability}.

%% file: coc.tex
We formally define the composition order cycle problem.

\textbf{Composition Order Cycle (COC):} A history \( H \) contains COC if it can be extended (by appending zero or more response events) to some complete history \( H' \), such that: 
\begin{enumerate}
\item{L4 (subhistory equivalent):} \( \forall x \in \Theta \), \emph{complete}\( (H'|x) \) is equivalent to some legal sequential subhistory \( S|x \).
\item{L5 (cyclic order):} Let \( <_x \) be the corresponding single-object order of \( S|x \). Let \( < \) be the transitive closure of the union of all \( <_x \) orders with \( <_{H|P} \) orders. There exists a set of operations \( op_1, op_2, \ldots, op_n \) such that a cyclic dependency exists with \( op_1 < op_2 < \ldots < op_n \), \( op_n < op_1 \), and each pair is directly related by some \( <_x \) or \( <_{H|P} \).
\end{enumerate}

%% file: correctness.tex
\begin{lemma}
\label{lemma:1}
OSC and COC are mutually exclusive. Moreover, a history \( H \) that satisfies
\( L4 \) and does not satisfy \( L5 \) is OSC.
\end{lemma}

\begin{proof}
We first show that OSC and COC are mutually exclusive.

If a history \( H \) is OSC, then by definition there exists a legal sequential
history \( S \) such that \( S \) satisfies \( L1 \), \( L2 \), and \( L3 \).
In particular, \( S \) induces a strict total order \( <_S \) that extends all
per-object orders \( <_x \) and all per-process orders \( <_{H|P} \).
Therefore, the union of these relations is acyclic, and no cyclic dependency
as described in \( L5 \) can exist. Hence, \( H \) does not contain COC.

Conversely, suppose that \( H \) contains COC. Then by \( L5 \), the transitive
closure of the union of all per-object orders \( <_x \) and per-process orders
\( <_{H|P} \) contains a cycle. No strict total order can extend a cyclic
relation, and thus no sequential history \( S \) can satisfy \( L1 \)--\( L3 \).
Therefore, \( H \) is not OSC.

We now show that if a history \( H \) satisfies \( L4 \) and does not satisfy
\( L5 \), then it is OSC. Since \( L5 \) does not hold, the transitive closure
of the union of all per-object orders \( <_x \) and per-process orders
\( <_{H|P} \) is acyclic, and thus forms a strict partial order over operations.
By the order-extension theorem, there exists a strict total order \( <_S \)
that extends this partial order. The corresponding sequential history \( S \)
is legal by \( L4 \), and preserves all per-process orders by construction,
thereby satisfying \( L1 \) and \( L2 \). Hence, \( H \) is OSC.
\end{proof}

\begin{lemma}
\label{lemma:2}
A history \( H \) does not contain a composition order cycle (COC) if all write operations are in a strict total order \( < \).
\end{lemma}

\begin{proof}
We prove by contradiction. Assume that history \( H \) contains a COC. By
Definition~L5, there exists a cycle in the transitive closure of the union of
all per-object orders \( <_x \) and per-process orders \( <_{H|P} \).

Observe that a cycle cannot consist solely of read operations, since read
operations do not impose ordering constraints on object states and thus cannot
introduce cyclic dependencies across object orders. Hence, any COC must include
at least one write operation.

Moreover, a cycle must include at least two distinct write operations.
Otherwise, if the cycle contained exactly one write operation, removing that
write would eliminate all object-level ordering constraints, leaving only
per-process orders, which are acyclic by definition. This contradicts the
existence of a cycle.

Therefore, any COC must contain at least two write operations. However, by
assumption, all write operations are totally ordered by a strict total order.
A strict total order is acyclic and antisymmetric, and thus cannot contain a
cycle. Hence, no such COC can exist in \( H \), yielding a contradiction.
\end{proof}

This lemma has been used in many sequential coordination services, such as ZooKeeper, which uses a linearizable write to ensure all reads are sequential. However, they did not formally prove it this way. Zoonet provides synchronization operations to make their system OSC, and they also offered a similar proof based on their mechanism.

\begin{invariant}
\label{invariant:1}
We require the following invariants to hold in our system for write-order broadcast.

\emph{Integrity.} If some process delivers a message \( m \), then there exists a
process \( P_i \in \Pi \) that has broadcast \( m \).

\emph{Strict Total Order} (\( <_{\text{cast}} \)).
If a process delivers message \( m \) before message \( m' \), then every process
that delivers \( m' \) must also deliver \( m \), and must deliver \( m \) before
\( m' \).
\end{invariant}

\begin{property}
\label{property:1}
For executions in which the object-to-node mapping in the CoPPar Tree remains
unchanged, the write-order broadcast ensures subhistory equivalence: for every
object \( x \in \Theta \), \( \emph{complete}(H'|x) \) is equivalent to some legal
sequential subhistory \( S|x \).
\end{property}

By Invariant~\ref{invariant:1}, all processes observe the same strict total order
of delivered write messages. For a fixed object \( x \), this induces a strict
total order over all write operations on \( x \). Together with the fact that
each process subhistory is sequential, this total order defines a legal
sequential subhistory \( S|x \) equivalent to \( \emph{complete}(H'|x) \).

\begin{property}
\label{property:2}
In any history \( H \) that satisfies \( L4 \) (subhistory equivalence),
CoPPar-Cen and CoPPar-Dec ensure that all write operations are ordered by a
global strict total order \( < \).
\end{property}

According to Invariant~\ref{invariant:1}, the write-order broadcast delivers all write operations in a strict total order \( <_{\text{cast}} \) at all processes. For two write operations
on the same object \( x \), the corresponding single-object order \( <_x \) and the
per-process order \( <_{H|P} \) are preserved within \( <_{\text{cast}} \). For two
write operations on different objects \( x \) and \( y \), since there is no direct
single-object order between them, their relative order in \( <_{H|P} \) is preserved
by the broadcast algorithm. Consequently, \( <_{\text{cast}} \) is a strict total order
that encompasses both all per-object orders \( <_x \) and all per-process orders
\( <_{H|P} \).

\begin{property}
\label{property:osc}
For executions in which the object-to-node mapping remains unchanged (i.e.,
processes do not perform change node operations), the CoPPar Tree guarantees OSC.
\end{property}

\begin{proof}
By Property~\ref{property:1}, in such executions the CoPPar Tree satisfies \( L4 \)
(subhistory equivalence). Meanwhile, by Property~\ref{property:2} and
Invariant~\ref{invariant:1}, all write operations are delivered in a strict total
order \( <_{\text{cast}} \) that preserves both per-object orders \( <_x \) and
per-process orders \( <_{H|P} \).

Let \( H \) denote the history of operations in the CoPPar Tree. Then, by
Lemma~\ref{lemma:2}, \( H \) does not contain a composition order cycle (COC).
Since \( H \) satisfies \( L4 \) and contains no COC, Lemma~\ref{lemma:1} implies
that the CoPPar Tree guarantees ordered sequential consistency (OSC).
\end{proof}

\begin{property}
\label{property:osc-change}
The CoPPar Tree guarantees OSC even when processes perform change node operations.
\end{property}

\begin{proof}
We prove this by induction on the number of change node operations in the history \( H \).

\textbf{Base case:} No change node operations. By Property~\ref{property:osc}, the CoPPar Tree guarantees OSC.

\textbf{Inductive step:} Assume that for \( k \) change node operations, the history
can be linearized to a sequential history \( S \) satisfying OSC. Consider adding the
\((k+1)\)-th change node operation \( op_{\text{change}} \) performed by process \( P_i \).

By the write-order broadcast mechanism and Invariant~\ref{invariant:1}, all writes,
including those due to upgrade or downgrade operations, are delivered in a strict
total order \( <_{\text{cast}} \) at all processes. The CoPPar Tree ensures that
causal dependencies between parent and child nodes are preserved (\( w^{\text{child}} <_c w^{\text{parent}} \)).

Therefore, inserting \( op_{\text{change}} \) does not violate the global strict total
order of writes. All per-process orders \( <_{H|P} \) and per-object orders \( <_x \)
are maintained. By Lemma~\ref{lemma:2}, no composition order cycle (COC) is formed.
Since \( L4 \) holds and no COC exists, Lemma~\ref{lemma:1} implies that OSC is
guaranteed after \( k+1 \) change node operations.

By induction, OSC holds for any finite number of change node operations.
\end{proof}

In summary, the CoPPar Tree guarantees OSC and eliminates the COC problem.